\documentclass[12pt]{article}

\usepackage{bigints}
\usepackage{amsfonts}
\usepackage{amsthm}
\usepackage{amsmath}
\usepackage{amssymb}

\swapnumbers

\newtheorem{Thm}{Theorem}[section]

\newtheorem{Lem}[Thm]{Lemma}

\newtheorem{Prop}[Thm]{Proposition}
\newtheorem*{Cara}{Carath\'eodory's Theorem} 
\newtheorem*{Thm2.1.3}{{Theorem 2.1.3 of \cite{CoddingtonTheory}}}
\newtheorem*{ThmI.5.3}{{Theorem I.5.3 of \cite{Hale2009Ordinary}}}

\theoremstyle{definition}
\newtheorem{Def}[Thm]{Definition}

\theoremstyle{remark}

\newcommand{\man}[1]{\ensuremath{\mathcal{#1}}} 
\newcommand{\ABound}[1][M]{\ensuremath{\mathcal{B}(\mathcal{#1})}} 
\newcommand{\bound}[2][]{\ensuremath{\partial{#1}(\man{#2})}} 
\newcommand{\covers}{\ensuremath{\rhd}} 
\newcommand{\EmptySet}{\ensuremath{\varnothing}}

\begin{document}

  \title{Generalizations of the {A}bstract {B}oundary singularity theorem}
  \author{Ben E Whale$^{124}$, Mike J S L Ashley$^{235}$, and Susan M Scott$^{26}$}

  \footnotetext[1]{Department of Mathematics and Statistics,
  Division of Sciences, University of Otago} 

  \footnotetext[2]{Centre for Gravitational Physics, 
    Department of Quantum Science, 
  College of Physical and Mathematical Sciences, 
  The Australian National University} 

  \footnotetext[3]{Financial Risk Management, KPMG Advisory Australia}

  \footnotetext[4]{ben.whale@otago.ac.nz}
  \footnotetext[5]{mashley1@kpmg.com.au}
  \footnotetext[6]{susan.scott@anu.edu.au}

  \maketitle

  \begin{abstract}
    The Abstract Boundary singularity theorem was first proven by Ashley and
    Scott. It links the existence of incomplete causal geodesics 
    in strongly causal, maximally extended spacetimes to the
    existence of Abstract Boundary
    essential singularities, i.e., non-removable singular
    boundary points. We give two 
    generalizations of this theorem: the first to continuous causal 
    curves and the
    distinguishing condition, the second to locally Lipschitz curves in
    manifolds such that no inextendible locally Lipschitz curve is
    totally imprisoned. To do this we extend generalized affine parameters
    from $C^1$ curves to locally Lipschitz curves.
  \end{abstract}

\section{Introduction}

  The end goal of our program of research is to link the 
  Penrose Hawking Singularity Theorems to
  curvature singularity results.
  Our three theorems, Theorems \ref{thm.abst_original},
  \ref{CurLim:Thm.NewAbstractBoundarySingularitTheorem} 
  and \ref{CurLim:Thm.NewAbstractBoundarySingularitTheoremLL},
  are a further step towards this goal. They prove that the Penrose
  Hawking Singularity Theorems actually 
  imply the existence of irremovable, also
  called essential, singularities and they provide a location for
  these singularities in terms of boundary points of an envelopment.
  It is our hope that this additional structure can be exploited to
  complete our program of research.

  The Abstract Boundary singularity theorem, proven by Ashley and Scott
  \cite[Theorem 4.12]{Ashley2002a}, is:
  \begin{Thm}
    \label{thm.abst_original}
    Let $(\man{M},g)$ be a strongly causal, $C^l$ maximally
    extended, $C^k$ spacetime $(1\leq l\leq k)$. 
    Let $\mathcal{C}$ be the 
    set of affinely parametrized causal geodesics in $\man{M}$.
    There
    exists an incomplete curve in $\mathcal{C}$ 
    if and only if the Abstract Boundary $\ABound$
    contains an abstract $C^l$ essential singularity.
  \end{Thm}
  An abstract $C^l$ essential singularity is an abstract boundary
  set which has a singleton $\{p\}$ as a representative boundary
  set where $p$ is
  a singular
  boundary point which cannot be
  removed by a change of coordinates and is approached by a curve in 
  $\mathcal{C}$  with bounded parameter.
  This theorem does not prove the existence of an incomplete
  causal
  geodesic, but rather shows that the existence of
  an incomplete causal geodesic is equivalent to the existence of an
  endpoint for the incomplete geodesic: that is, a location
  for the singularity in the Abstract Boundary.   
  Hence, the theorem extends the `standard' singularity theorems, e.g.,
  the Penrose and Hawking singularity theorems 
  \cite[Section 8.2]{HawkingEllis1973}, by showing that
  they actually produce genuine singularities,
  at least according to the Abstract Boundary classification of boundary
  points \cite[Theorem 4.13]{Ashley2002a}.
  For further details about the Abstract Boundary
  please refer to one of 
  \cite{Ashley2002a,ScottSzekeres1994,Whale2010,Whale2014Chart}.

  Ideally, the use of geodesics and the 
  assumption of strong causality could be relaxed to increase the 
  generality of the theorem. Geroch \cite{Geroch1968b} has shown that
  in order to identify all singular behaviour in a spacetime it is
  necessary to consider, at least, all causal curves. Hence,
  it is desirable that 
  the singularity theorem above be generalized to include, at least,
  all 
  causal curves.
  
  The most general singularity theorems,
  like that given by Maeda and Ishibashi \cite{MaedaIshibashi1996}, use 
  causality conditions much weaker than strong causality. Ashley
  and Scott
  have investigated weakening the strong causality condition 
  in some cases.  In particular, they have shown that
  the Abstract Boundary singularity theorem is false in chronological 
  spacetimes and have indicated how a counter example may be provided in
  causal spacetimes \cite[Section 4.4.1]{Ashley2002a}. In 
  \cite[Theorem 4.22]{Ashley2002a} Ashley 
  and Scott also prove that in 2-dimensional spacetimes the theorem holds
  for the distinguishing condition.
  From a theoretical point of view it is also pleasing to find the weakest
  conditions under which the theorem holds.

  The main theorem of this paper is the following, which we suggest
  should be considered as \emph{the} Abstract Boundary singularity
  theorem.

  \begin{Thm}[The Abstract Boundary Singularity Theorem]
    \label{CurLim:Thm.NewAbstractBoundarySingularitTheorem}
    Let $(\man{M},g)$ be a future (past) distinguishing, $C^l$ maximally
    extended, $C^k$ spacetime $(1\leq l\leq k)$ and let $\mathcal{C}$
    be the family of generalized affinely parametrized continuous causal 
    curves in $\man{M}$.
    There
    exists an incomplete curve in $\mathcal{C}$ 
    if and only if $\ABound$
    contains an abstract $C^l$ essential singularity.
  \end{Thm}

  There are two main difficulties with proving this theorem.
  The first is that we need a definition of a
  generalized affine parameter on a continuous causal curve 
  (Definition \ref{continuousGAP}) as 
  generalized affine parameters are usually defined on
  $C^1$ curves \cite[Page 208]{BeemEhrlichEasley1996}.
  The second is that we need the existence
  of an endpoint for a continuous causal curve
  in order
  to imply that
  the generalized affine parameter is bounded
  and
  the curve is extendible (Proposition \ref{prop:llextend}).
  To tackle these difficulties we work with a more general class of curves, 
  locally Lipschitz curves, Definition \ref{def:locallylipschitz}
  (although some additional work is required
  to apply results about locally Lipschitz curves
  to continuous causal curves). In proving the required results for
  locally Lipschitz curves we get, in addition to
  Theorem \ref{CurLim:Thm.NewAbstractBoundarySingularitTheorem},
  the following theorem.

  \begin{Thm}\label{CurLim:Thm.NewAbstractBoundarySingularitTheoremLL}
    Let $(\man{M},g)$ be a $C^l$ maximally
    extended, $C^k$ spacetime $(1\leq l\leq k)$ so that no inextendible
    locally
    Lipschitz curve
    is totally imprisoned.
    Let $\mathcal{C}$
    be the family of generalized affinely parametrized locally Lipschitz 
    curves in $\man{M}$.
    There
    exists an incomplete curve in $\mathcal{C}$ 
    if and only if $\ABound$
    contains an abstract $C^l$ essential singularity.
  \end{Thm}

  While Theorem \ref{CurLim:Thm.NewAbstractBoundarySingularitTheoremLL}
  is more general than Theorem 
  \ref{CurLim:Thm.NewAbstractBoundarySingularitTheorem} it lacks 
  physical motivation for the condition on $\man{M}$.

  Consider the three statements in Theorem
  \ref{CurLim:Thm.NewAbstractBoundarySingularitTheoremLL}:
  \begin{description}
    \item[(1)] The spacetime does not contain an inextendible totally imprisoned locally Lipschitz
    curve,
    \item[(2)] The set of curves, $\mathcal{C}$, contains an incomplete curve,
    \item[(3)] The Abstract Boundary contains an abstract $C^l$ essential
    singularity.
  \end{description}
  The following implications are apparent from the definitions and the
  proof of Theorem \ref{CurLim:Thm.NewAbstractBoundarySingularitTheoremLL}:
  \begin{enumerate}
    \item (1) implies that (2) and (3) are equivalent,
    \item (2) implies that either (1) does not hold or (3) holds,
    \item (3) implies that (2) holds and that this incomplete curve
      is inextendible and not totally imprisoned.
  \end{enumerate}
  Hence, given the use of locally Lipschitz curves,
  Theorem \ref{CurLim:Thm.NewAbstractBoundarySingularitTheoremLL} cannot
  be weakened further. 
  Theorem \ref{CurLim:Thm.NewAbstractBoundarySingularitTheoremLL}
  can, therefore, be considered a proof of the generally accepted statement
  that incompleteness of curves implies either the existence of incomplete
  trapped curves or of a
  singularity.

  More generally, this paper fits into the wider field of research
  into singularities and
  the use and application of boundary constructions in General 
  Relativity. We will not provide further context here, but we refer
  the interested reader to Senovilla's series of review articles on
  singularity theorems 
  \cite{Senovilla1998,2014arXiv1410.5226S,2006physics...5007S},
  recent work on the causal boundary \cite{Flores2010Final} and 
  Geroch's $g$-boundary \cite{Bentabol2014} as well as newer
  boundary constructions 
  \cite{Whale2014Chart,GarciaParrado:2005yz,citeulike:13467606}.
  With regards to the Abstract Boundary we refer the reader to
  \cite{Whale2014Chart,citeulike:9599226,citeulike:13467610} for
  discussions of topological properties,
  \cite{Whale2014Chart,citeulike:8625460} for its relationship
  to boundaries induced by charts and distances,
  \cite{FamaScott1994,FamaClarke1998} for homotopy and rigidity results,
  \cite{Ashley2002a,Ashley2002b,AshleyScott2003} for other Abstract Boundary
  singularity results and
  \cite{Ashley2002a,Whale2010} for general reviews of the Abstract Boundary.

\section{Common definitions}

  Throughout this paper we restrict our attention to $n$-dimensional, 
  paracompact,
  connected, Hausdorff, $C^\infty$ manifolds $\man{M}$ which are Lorentzian
  and time orientable.
  Standard definitions, e.g., of strong causality, the maximal extension of a
  manifold and of the generalized affine parameter, are taken from 
  \cite{BeemEhrlichEasley1996}. We do, however, deviate from
  \cite{BeemEhrlichEasley1996} by using the
   `total imprisonment' of \cite{HawkingEllis1973}
  rather than the 
  `imprisonment' of \cite{BeemEhrlichEasley1996}. 
  The two definitions are equivalent; we simply prefer the
  terminology of \cite{HawkingEllis1973}.
  We do not review the Abstract Boundary, though we encourage the reader
  to refer to one of 
  \cite{Ashley2002a,ScottSzekeres1994,Whale2010}.
  Because curves and their extensions play a central role in this paper, we
  remind the reader of the following definitions.

  \begin{Def}
    A $C^k$, $k\geq 1$, 
    (regular) curve 
    $\gamma:[a,b)\to \man{M}$, $a,b\in\mathbb{R}\cup\{\infty\}$, $a<b$, 
    is a $C^k$ function
    with
    $\gamma'$ everywhere
    non-zero. We say that $\gamma$ is
    non-spacelike (timelike) and future (past) directed if
    $\gamma'$ is everywhere non-spacelike (timelike) and 
    future (past) directed. 
  \end{Def}

  We shall need the following non-standard definition.

  \begin{Def}\label{def:locallylipschitz}
    A (regular) 
    locally Lipschitz curve $\gamma:[a,b)\to\man{M}$ is a function
    so that, for each chart $\phi:U\subset \man{M}\to\mathbb{R}^n$
    and each $t\in[a,b)$ such that $\gamma(t)\in U$, there exists
    $V$, a neighbourhood of $t$ in $[a,b)$, and
    $K\in\mathbb{R}^+$ such that $\gamma(V)\subset U$ and for all $t_1,t_2\in V$,
    \[
      d(\phi\circ\gamma(t_1),\phi\circ\gamma(t_2))\leq K\lvert t_1 - t_2\rvert,
    \]
    where $d$ is the Euclidean distance on $\mathbb{R}^n$ and so that
    $\gamma'$ is non-zero apart from a set of measure zero.
    Note that $K$ depends on the chart $\phi$ and the point $t$.
  \end{Def}

    The definition is independent of the choice of chart since 
    the set $V$ can be taken to be compact and
    changes of coordinates between charts are invertible and bounded on 
    compact sets contained in the intersection of their domains.

  \begin{Def}[{\cite[Section 3.2]{BeemEhrlichEasley1996}}
  or {\cite[Page 184]{HawkingEllis1973}}]
    \label{def.ccurve}
    A continuous future directed, non-spacelike curve 
    $\gamma:[a,b)\to\man{M}$ is a 
    continuous function so that for each $t_0\in [a,b)$ there is a
    neighbourhood $N$ of $t_0$ in $[a,b)$
    and a convex normal neighbourhood $U$ of $\gamma(t_0)$ so that
    for all $t\in N$, $t\neq t_0$, if $t>t_0$ then 
    $\gamma(t)\in J^+(\gamma(t_0),U)- \gamma(t_0)$ or
    if $t<t_0$ then $\gamma(t)\in J^-(\gamma(t_0),U)-\gamma(t_0)$. 
    We say that $\gamma$ is past directed if,
    for $t>t_0$, then $\gamma(t)\in J^-(\gamma(t_0),U)- \gamma(t_0)$ and
    for $t<t_0$ then $\gamma(t)\in J^+(\gamma(t_0),U)-\gamma(t_0)$.
    We say that $\gamma$ is timelike
    if the sets $I^+(\gamma(t_0),U)$ 
    and $I^-(\gamma(t_0),U)$ are used 
    instead of $J^+(\gamma(t_0),U)- \gamma(t_0)$ 
    and $J^-(\gamma(t_0),U)- \gamma(t_0)$ respectively.
    We shall assume that every continuous causal (timelike) curve
    is equipped with a parametrization so that it is also a locally
    Lipschitz curve.
    That such a parametrization always exists is proven
    on page 75 of \cite{BeemEhrlichEasley1996}. 
  \end{Def}

  It is worth noting that the idea that every continuous causal
  curve is also locally Lipschitz makes up part of the folklore of
  General Relativity. To the best of
  the authors' knowledge the first mention of this result, where it is stated
  without proof, is by Penrose \cite[Remark 2.26]{Penrose1972} who credits
  it to Geroch but gives no reference. While it is certainly true that
  every continuous causal curve has a reparametrization so that it is 
  locally Lipschitz, it is possible to give continuous causal curves
  parameters so that they are not locally Lipschitz (e.g. $t\mapsto t^{1/3}$).

  The condition that $\gamma(t)\neq\gamma(t_0)$, for all
  $t\in N - \{t_0\}$, is the continuous causal curve
  analogue of the non-zero tangent
  vector condition for $C^k$, $k>0$, curves.

  \begin{Def}[{\cite[Definition 2]{ScottSzekeres1994}}]
    A curve $\gamma:[a,b)\to\man{M}$
    is a subcurve of a curve 
    $\lambda:[a',b')\to\man{M}$,
    if $a'\leq a<b\leq b'$ and $\lambda|_{[a,b)}=\gamma$. 
    If $a=a'$ and $b<b'$ then $\lambda$ is an extension
    of $\gamma$.
  \end{Def}

  \begin{Def}[Compare to {\cite[Definition 6.2]{BeemEhrlichEasley1996}}]
    \label{Def:incomplete_geo}
    An affinely parametrized causal geodesic $\gamma:[a,b)\to \man{M}$ is
    incomplete if $b<\infty$ and $\gamma$ is not extendible by any
    affinely parametrized
    causal geodesic.
  \end{Def}

  \begin{Def}
    Two $C^k$, $k\geq 0$, 
    (locally Lipschitz) curves $\gamma:[a,b)\to \man{M}$ and
    $\lambda:[a',b')\to\man{M}$ are related by a change of parameter
    if there exists a $C^k$ (locally Lipschitz) surjective 
    strictly monotonically
    increasing function
    $f:[a',b')\to[a,b)$ so that $\lambda = \gamma\circ f$.
  \end{Def}

  Note that, unlike
  $C^k$ changes of parameter, the class of locally Lipschitz
  curves is not closed under locally Lipschitz reparametrizations.
  The class of locally Lipschitz curves is
  closed, however, under locally bi-Lipschitz reparametrizations,
  though we do not make that restriction here.

  The next definition is non-standard, however it 
  simplifies the discussion of many of the results in this paper.

  \begin{Def}
    Let $\gamma:[a,b)\to\man{M}$ 
    be a $C^k$, $k\geq 0$, curve. A full sequence in $\gamma$
    is a sequence $\{x_i=\gamma(t_i)\}_{i\in\mathbb{N}}$, 
    $\{t_i\}\subset [a,b)$, 
    $t_i<t_{i+1}$, so that $t_i\to b$
    as $i\to \infty$.
    We say that $x\in\man{M}$ is a limit point of $\gamma$ if there exists
    a full sequence $\{x_i\}$ in $\gamma$ so that $x_i\to x$. 
    We write $\gamma\to x$
    if and only if every full sequence in $\gamma$ converges to $x$, 
    in which case 
    $x$ is the endpoint of $\gamma$. 
    We say that a curve $\gamma$ is a winding curve
    if there exist two full sequences in $\gamma$ with different limit points.
  \end{Def}

  This terminology
  is inspired by the Misner spacetime \cite[Section 5.8]{HawkingEllis1973}
  in which every curve with at least two limit
  points `winds' around the cylinder.

  \begin{Def}
    A curve, $\gamma$, is precompact if the closure of its image,
    $\overline{\gamma}$, is compact.
  \end{Def}

\section{The Abstract Boundary singularity theorem}

  Since details of the Abstract Boundary singularity theorem
  have only appeared in Ashley's PhD Thesis
  \cite[Theorem 4.12]{Ashley2002a} we
  present here Ashley and Scott's proof of Theorem \ref{thm.abst_original}.
  We do not review the end-point Theorem
  \cite[Theorem 3.2.1]{Whale2010} which plays an important
  part in the proof below (Ashley's Thesis
  \cite[Section 4.2]{Ashley2002a} contains a nice discussion of the
  end-point Theorem but does not give its proof).
  Note that the $\Leftarrow$ implication of this
  proof follows
  simply by definition and the Hausdorff property of manifolds.
  It does not require a restriction on the causality of the spacetime
  $(\man{M},g)$.

  \begin{proof}[Proof of Theorem {\ref{thm.abst_original}}]
    $\Leftarrow$
	    Let $[p]\in\ABound[\man{M}]$ be an abstract $C^l$ essential singularity.
	    That is, there exists $\mu_p:\man{M}\to\man{M}_p$ an
	    envelopment so that $p\in\partial\mu_p(\man M)$ is a $C^l$
	    essential singularity.  Thus there exists
	    $\gamma:[a,b)\to \man{M}\in\mathcal{C}$, $b<\infty$, so that
	    $p$ is a limit point of $\mu_p(\gamma)$.
	    Let $\{x_i=\gamma(t_i)\}_{i\in\mathbb{N}}$ 
      be a full sequence in
	    $\gamma$ so that $\mu_p (x_i)$ converges to $p$.
	
	    Suppose that
	    there exists $\lambda:[a,c)\to\man{M}$ an extension
	    of $\gamma$, where $\lambda\in\mathcal{C}$. 
      Consider the sequence $\{y_i=\lambda(t_i)\}$. 
	    Since $\lambda$ is an extension of $\gamma$, $x_i=y_i$ so
	    that $\{\mu_p (y_i)\}$ converges to $p$. Yet, $t_i\to b$ and
	    $b<c$ so $y_i\to\lambda(b)$. Since $\man{M}_p$ is
	    Hausdorff and $\mu_p$ is continuous
      we see that $\mu_p(\lambda(b))=p$ which is a contradiction
	    since $\mu_p(\lambda(b))\in\mu_p(\man{M})$ 
      and $p\in\bound[\mu_p]{\man{M}}$. Therefore $\gamma$ is
	    an incomplete curve in $\mathcal{C}$ as required.
    
    $\Rightarrow$ Let
	    $\gamma\in\mathcal{C}$ be incomplete.
	     We
	    have two cases:

    \begin{description}
      \item[\textbf{Case 1}]
        Suppose that there exists
        a full sequence $\{x_i\}_{i\in\mathbb{N}}$ in 
        $\gamma$ so that $\{x_i\}$ has 
        no limit points in $\man{M}$.
      	
      	By the end-point Theorem \cite[Theorem 3.2.1]{Whale2010} 
        there exists an
        envelopment $\phi:\man{M}\to\man{M}_\phi$ so that
        $\{\phi(x_i)\}$ converges to $x\in\bound[\phi]{\man{M}}$. We will
        now proceed to classify $x$ according to the classification of 
        boundary points given in \cite{ScottSzekeres1994}.

        Since
        $(\man{M}, g)$ is $C^l$ maximally extended, $x$ cannot be a
        $C^l$ regular boundary point. By assumption, $x$ is
        approached by $\phi(\gamma)$ ($\gamma$ is 
        incomplete). Hence we can
        conclude that
        $x$ is a $C^l$ singularity. Suppose that $x$ is a $C^l$ removable
        singularity, then by Theorem 43 of \cite{ScottSzekeres1994} 
        there exists a boundary
        set $B$, of another envelopment, so that $B\covers x$ and $B$ 
        contains at least one
        $C^l$ regular boundary point. This contradicts that $(\man{M}, g)$ is
        $C^l$ maximally extended.
        Thus $x$ is a $C^l$ essential singularity. This implies
        that $[x]\in\ABound[\man{M}]$ is an abstract $C^l$ essential
        singularity.
      	
      \item[\textbf{Case 2}] 
        Suppose that every full sequence in $\gamma$ has a limit point
        in $\man{M}$.
        We distinguish two further cases.

        \begin{description}
          \item[\textbf{Case 2.1}]
            Suppose that every full sequence in $\gamma$ 
            has the same limit point. Thus there exists $p\in \man{M}$ 
            so that $\gamma\to p$. 
            
            Let $\gamma:[a,b)\to \man{M}$, where $b<\infty$, and
            let $N$ be a convex normal neighbourhood of $p$. Since $\gamma\to p$
            there exists $t\in[a,b)$ so that $\gamma|_{[t,b)}\subset N$.
            By the convexity of $N$ there exists $v\in T_p\man{M}$ so that
            $\gamma(t)=\exp_p(v)$. 
            Since $\gamma$ is a geodesic this implies that there exists 
            $\alpha\in\mathbb{R}^+$
            so that $\lim_{\tau\to b}\gamma'(\tau)=-\alpha v$.
            Let $\lambda:[0,c)\to\man{M}$ 
            be the unique affinely parametrized causal 
            geodesic from $p$ with
            tangent vector $-\alpha v$ at $p$, where $c$ is
            chosen so that $\lambda\subset N$. 
            The convexity of $N$ implies
            that the curve, $\mu$, 
            given by the concatenation of
            $\gamma$ and $\lambda$ is, itself, a causal geodesic. 
            Standard properties of affine parameters 
            (see the last paragraph of page 17 of
            \cite{BeemEhrlichEasley1996}) imply that
            we can affinely parametrize $\mu$ so that
            $\mu:[a,b+c)\to \man{M}$, 
            $\gamma=\mu|_{[a,b)}$ and $\lambda(t) = \mu(b+t)$.
            Since $[a,b)\subset[a,b+c)$, 
            the causal geodesic $\mu$ is an extension of $\gamma$.
            Hence $\gamma$ is not incomplete.
            
            This is a contradiction
            and thus this case cannot occur.
          \item[\textbf{Case 2.2}]
            Suppose that there exist two full sequences in $\gamma$
            with different limit points.

            Let $\gamma:[a,b)\to \man{M}$, where $b<\infty$, and let
            $\{x_i=\gamma(t_i^x)\}_{i\in\mathbb{N}}$ and
            $\{y_i=\gamma(t_i^y)\}_{i\in\mathbb{N}}$ 
            be the two full sequences
            with limit points $x,y\in \man{M}$, so that $x_i\to x$, 
            $y_i\to y$ and $x\neq y$. Without loss of generality,
            we can assume that for all $i$, $t_i^x<t_i^y<t_{i+1}^x$.

            Let $U,V\subset \man M$ be open sets so that $x\in U$,
            $y\in V$ and $U\cap V=\EmptySet$. Since $x_i\to x$ there exists
            $N_x\in\mathbb{N}$
            so that for all $i\geq N_x$, $x_i\in U$. Similarly, 
            as $y_i\to y$ there exists $N_y\in\mathbb{N}$
            so that for all $i\geq N_y$, $y_i\in V$.
            Let $i\geq\max\{N_x,N_y\}$, then $\gamma(t_i^x)\in U$,
            $\gamma(t_i^y)\in V$ and $\gamma(t_{i+1}^x)\in U$. Since
            $t_i^x<t_i^y<t_{i+1}^x$, $x\neq y$ and as $U$ and $V$ 
            are arbitrary
            the
            spacetime $(\man M, g)$ is not strongly causal.

            This is a
            contradiction and therefore this case cannot occur.
        \end{description}
    \end{description}

    As only case 1 may occur, we have proven our
    result.
  \end{proof}

  Note that the proof of the Abstract Boundary singularity theorem is based
  on the division of causal geodesics into three classes:
  \begin{enumerate}
    \item curves with a full sequence with no limit points, i.e.\
    non-precompact curves,
    \item curves with all full sequences having the same limit point, i.e.\
    precompact curves with an endpoint,
    \item curves with all full sequences having a limit point and with
    two full sequences with different limit points, i.e.\
    precompact curves without an endpoint (these curves are necessarily
    winding).
  \end{enumerate}
  The first class of curves correspond to Abstract Boundary essential
  singularities. The second class of curves correspond to curves that are
  not incomplete.
  The third class of curves correspond to curves that violate strong causality.
  Thus to generalize this theorem to all continuous causal curves and
  the distinguishing condition we need three
  things:
  \begin{enumerate}
    \item an incomplete continuous causal curve must correspond to an abstract
      essential singularity,
    \item a continuous causal curve with an endpoint must
      not be incomplete,
    \item the existence of a precompact winding continuous causal 
    curve must violate the distinguishing
    condition.
  \end{enumerate}

  The first point requires that each continuous causal curve carries a
  particular
  parametrization so that the set of all continuous causal curves satisfies the
  bounded parameter property. For more background on this
  we refer the reader to a discussion of the classification of the
  Abstract Boundary; see one of
  \cite{Ashley2002a,ScottSzekeres1994,Whale2010}.
  \begin{Def}[{\cite[Definition 4]{ScottSzekeres1994}}]
    \label{Def.bbp}
    Let $\mathcal{C}$ be a set of curves in $\man{M}$. The set $\mathcal{C}$ has
    the bounded parameter property if:
    \begin{enumerate}
      \item for all $p\in \man{M}$ there exists $\gamma\in \mathcal{C}$ so that
      $p\in \gamma$,
      \item if $\gamma\in \mathcal{C}$ then every subcurve of $\gamma$ is
      in $\mathcal{C}$,
      \item for all $\gamma,\lambda\in\mathcal{C}$ if there exists
      a change of parameter relating $\gamma$ and $\lambda$ then
      either both parameters are bounded or both are unbounded.
    \end{enumerate}
  \end{Def}
  The set of all affinely parametrized causal geodesics satisfies the bounded
  parameter property. We will show that locally Lipschitz curves carry a
  generalization of the generalized affine parameter defined on $C^1$ curves
  \cite[Page 208]{BeemEhrlichEasley1996} and
  that this parametrization ensures that the set of all locally Lipschitz
  curves satisfies the bounded parameter property. 
  Since continuous causal
  curves are locally Lipschitz 
  it will then be the case
  that the set of
  all continuous causal curves satisfies the bounded parameter property.

  The second point requires a definition of an incomplete continuous
  causal curve. We use the new parameter to 
  define an incomplete locally Lipschitz curve, in the same spirit as 
  Definition \ref{Def:incomplete_geo}. We then show that a locally
  Lipschitz curve with endpoint is extendible, i.e., that the curve is
  not incomplete.

  For the third point we rely on a result from Hawking and Ellis, Proposition
  6.4.8 of \cite{HawkingEllis1973}, to get the needed contradiction.
  Paraphrased to accommodate our definitions of extension
  and total imprisonment \cite[Definition 7.29]{BeemEhrlichEasley1996} (note
  that we differ here from 
  Beem, Ehrlich and Easley as we use the
  `total imprisonment' of Hawking and Ellis which is 
  the 
  `imprisonment' of Beem, Ehrlich and Easley)
  the result is as follows:
  \begin{Prop}[{\cite[Proposition 6.4.8]{HawkingEllis1973}}]
    \label{prop.hawkellis}
    If the future or past distinguishing condition holds on a compact set $S$,
    there can be no inextendible causal curve totally
    imprisoned in $S$.
  \end{Prop}

  Note that the definition of causal curve {\cite[Page
  184]{HawkingEllis1973}} used by Hawking and Ellis is the same as ours,
  Definition \ref{def.ccurve}; thus Proposition \ref{prop.hawkellis}
  can be applied in our situation.

\section{The bounded parameter property and incompleteness for 
locally Lipschitz curves}

  \subsection{The bounded parameter property for locally Lipschitz curves}

  A generalized affine parameter
  is given as the arc length of a curve, with respect to a Riemannian metric,
  induced by an orthonormal frame which is parallelly propagated along the
  curve; see \cite[Page 208]{BeemEhrlichEasley1996}. Because of the need for
  parallel propagation, generalized affine parameters are usually defined on
  $C^1$ curves. There are, however, existence and uniqueness results for
  ordinary differential equations involving functions with weaker regularity
  than $C^1$. We exploit these results to define parallel propagation along
  locally Lipschitz curves and thus to equip them with a generalized 
  affine parameter.
  We go through the proof in detail.

  \begin{Prop}\label{continuous parallel propagation}
   Let $\gamma:[a,b)\to\man M$ be a locally Lipschitz curve. 
   Then for all $\tau\in [a,b)$ (excluding a set of Lebesgue measure zero)
   and all $v\in T_{\gamma(\tau)}\man M$ there
   exists $V:[a,b)\to T\man M$, a vector field 
   on $\gamma$, which is absolutely continuous on
   compact subsets of $[a,b)$ and so that $V(\tau)=v$ and 
   $\nabla_{\gamma'}V=0$ except on a set of Lebesgue measure zero.
  \end{Prop}
  \begin{proof}
    Since $\gamma$ is locally Lipschitz, $\gamma'(t)$ exists for almost 
    all $t\in[a,b)$;
    see \cite[Page 75]{BeemEhrlichEasley1996}.
    The equations for parallel propagation of a vector, $X$,
    along $\gamma$ in 
    some coordinate chart $\phi$
    are
    \[
    \frac{d}{dt}X^i(t)+(\gamma')^j(t)X^k(t)\Gamma_{jk}^i(t)=0.
    \]
    We will apply Carath\'eodory's Theorem 
    \cite[Theorem 2.1.1]{CoddingtonTheory} to this equation to prove the 
    existence of $X^i$ on some compact connected subinterval $I_\phi\subset [a,b)$, 
    such that
    ${\gamma(I_\phi)}$ is contained in the chart $\phi$. 
    As initial conditions we take
    $\tau\in \text{interior}(I_\phi)$ so that $\gamma'(\tau)$ exists and
    $v\in T_{\gamma(\tau)}\man{M}$.
    
    With this purpose in mind, let $\partial_i$, $i=1,\ldots,n$,
    be the coordinate vectors
    for the chart $\phi$. We can view $(\gamma')^j$ and $\Gamma_{jk}^i$
    as functions from $I_\phi$ to $\mathbb{R}$.
    Define $\hat{\gamma}^i:I_\phi\to\mathbb{R}$ by
    \[
      \hat{\gamma}^i(t)=\left\{\begin{aligned}
        (\gamma')^i(t) &&\quad \text{if}\ (\gamma')^i(t)\ \text{is defined}\\
        0 &&\quad \text{otherwise}.
      \end{aligned}\right.
    \]
    
    The vector $v\in T_{\gamma(\tau)}\man{M}$
    can be written as $v=v^i\partial_i$ where $v^i\in\mathbb{R}$.
    We consider $(\tau,v^1,\ldots,v^n)$ as a point
    in $\mathbb{R}^{n+1}$. 
    Let $|\cdot|$ be the $L_1$ norm on $\mathbb{R}^n$
    given by $|(x^1,\ldots,x^n)|=\sum_i|x^i|$. 
    We will write $x$ for $(x^1,\ldots,x^n)$. Thus 
    $\lvert x(t)\rvert= |(x^1(t),\ldots,x^n(t))|$.
    
    Choose $c>\max\{|v|,1\}$ an otherwise arbitrary positive constant. Let
    $R=\text{interior}\left(I_\phi\right)\times\{x\in\mathbb{R}^{n}:|x|<c\}$. 
    Thus $(\tau,v)\in R$.
    Letting $(t,x)\in\mathbb{R}^{n+1}$,
    we can define the functions $f^i:R\to\mathbb{R}$,
    $i=1,\ldots,n$, by
    \[
      f^i(t,x^1,\ldots,x^n)=\sum_{j,k}\hat{\gamma}^j(t)x^k\Gamma_{jk}^i(t).
    \]
    Since $\gamma$ is Lipschitz on $I_\phi$, we can see that $f^i$ is 
    continuous in $x^k$, $k =1,\ldots,n$, 
    for all fixed $t$ and 
    that $f^i$ is Lebesgue measurable in $t$ for all fixed $x^1,\ldots,x^n$.
    We will write $f(t,x)$ for $(f^1(t,x),\ldots,f^n(t,x))$.      
    
    By the Lipschitz condition on $\gamma$ and as each $\Gamma_{jk}^i$
    is continuous we know that the function 
    $m:I_\phi\to\mathbb{R}$ given by
    \[
      m(t)=c\sum_{i,j,k}|\hat{\gamma}^j(t)\Gamma_{jk}^i(t)|
    \]
    is Lebesgue integrable.      
    In particular, noting that $|x+y|\leq |x|+|y|$,
    some algebra shows that for all 
    $(t,x)\in R$ we have that
    $
      |f(t,x)|<m(t).
    $
    
    Carath\'eodory's Theorem 
    \cite[Theorem 2.1.1]{CoddingtonTheory}, rephrased for our particular
    situation, states that:
    \begin{Cara}
      Let $f^i$, $i=1,\ldots,n$, 
      be defined on $R$, 
      and suppose each $f^i$ is Lebesgue measurable in $t$
      for each fixed $x$ and continuous in $x$ for each fixed $t$.
      If there exists a Lebesgue integrable function 
      $m:I_\phi\to\mathbb{R}$ so that 
      $|f(t,x)|<m(t)$,
      then there exists
      $J\subset I_\phi$, 
      a subinterval of $I_\phi$, so that $\tau\in J$ and
      for each $i=1,\ldots,n$ there exists an absolutely continuous function
      $y^i(t):J\to\mathbb{R}$ so that
      $(t,y^1(t),\ldots,y^n(t))\in R$, $y^i(\tau)=v^i$, and so that
      \begin{equation}\label{ODE}
        \frac{d}{dt}y^i(t)+f^i(t,y^1(t),\ldots,y^n(t))=0
      \end{equation}
      for almost all $t$.
    \end{Cara}
    
    Since the functions $f$ and $m$ satisfy the conditions of the 
    theorem we know that on some subinterval $J$ of $I_\phi$, with $\tau\in J$,
    an absolutely continuous solution $X^i:J\to\mathbb{R}$, $i=1,\ldots,n$, 
    exists 
    except on a set of Lebesgue measure zero. Since $J\subset I_\phi$ could be
    taken to be very small, we only know the local existence of a solution
    about $\tau$.  
    
    To prove global existence on $I_\phi$ it is necessary to find an upper
    bound on the length of solutions with initial condition $v$ at $\tau$.
    Since each $\Gamma^i_{jk}$ is defined on all of $\phi$, as
    $\gamma$ is Lipschitz on $I_\phi$ and as $I_\phi$ is compact there
    exists $C>0$ so that 
    $C>\max_{k,t}\{\sum_{i,j}\lvert\hat{\gamma}^j(t)\Gamma_{jk}^i(t)\rvert\}$.
    Suppose that $X(t)$ is a local solution about $\tau$ so that 
    $X(\tau)=v$.
    We calculate that
    \begin{align*}
      \frac{d}{dt}\lvert X(t)\rvert &=\sum_{i}\frac{d}{dt}\lvert X^i(t)\rvert
          \leq \sum_{i}\lvert \frac{d}{dt} X^i(t)\rvert\\
          &= \sum_{i,j,k}
            \lvert \hat{\gamma}^j(t) X^k(t) \Gamma_{jk}^i(t) \rvert\\
          &\leq C \sum_k\lvert X^k(t)\rvert
          = C \lvert X(t)\rvert.
    \end{align*}  
    Gronwall's inequality \cite[Page 624]{Evans2010} now implies that
    there exists $K\in\mathbb{R}^+$ such that for all $t\in I_\phi$,
    \[
      |X(t)| \leq K\exp\left(Ct\right).
    \]
    Since $I_\phi$ is a compact subinterval of $I=[a,b)$, the function
    $\exp(Ct)$ has a maximum value on $I_\phi$. Hence, we know that if $X(t)$
    is a solution, then there exists $b\in\mathbb{R}^+$ so that
    $|X(t)|<b$ on $I_\phi$.   
    Since $c$ was arbitrary we can choose $c$ so that 
    $c>\max\{|v|,1,b\}$.
    
    Theorem 2.1.3 of \cite{CoddingtonTheory} now allows us to 
    conclude global existence; rephrased for our situation it
    states:
    \begin{Thm2.1.3}
      Let $R$ be an open, connected subset of $\mathbb{R}^{n+1}$,
      with points $(t,x)$, so that
      for each $i=1,\ldots,n$, the
      function $f^i$ is defined on $R$, Lebesgue measurable in $t$ for 
      each fixed $x$ and
      continuous in $x$ for each fixed $t$. If there exists
      a Lebesgue integrable function $m(t)$ so that
      $|f(t,x)|<m(t)$ for all $(t,x)\in R$, then any solution 
      $y^i$ of the system \eqref{ODE}
      in the sense of Carath\'eodory's Theorem can be extended to the 
      boundary of $R$.
    \end{Thm2.1.3}
    Since $f$ and $m(t)$ satisfy the conditions
    of the theorem, we know that the solution, 
    $X(t)$, extends to the boundary of $R$.
    The boundary of $R$ is 
    \[
      \text{interior}(I_\phi)\times\{x\in \mathbb{R}^{n}: |x|=c\}\cup
      \partial I_\phi\times\{x\in\mathbb{R}^n: |x|\leq c\},
    \]
    where $\partial I_\phi$ is the boundary of $I_\phi$ in $\mathbb{R}$.
    By construction we know that $|X(t)|<b$, hence, as
    $c>\max\{|v|,1,b\}$, the intersection of the 
    solution $(t, X^1(t),\ldots, X^n(t))$ with the
    boundary of $R$ cannot lie in  
    $\text{interior}(I_\phi)\times\{x\in \mathbb{R}^{n}: |x|=c\}$.
    Therefore 
    the solution must extend to 
    $\partial I_\phi\times\{x\in\mathbb{R}^n: |x|<c\}$. This implies
    that the solution is defined on
    all of $I_\phi$.

    To prove uniqueness we use Theorem I.5.3 of \cite{Hale2009Ordinary}.
    This theorem is a special case of 
    Theorem 2.2.1 of \cite{CoddingtonTheory} which better fits our
    situation.
    \begin{ThmI.5.3}
      Suppose that $R$ is an open set in $\mathbb{R}^{n+1}$, with points
      $(t,x)$,
      and for each $i=1,\ldots, n$,
      $f^i$ is defined on $R$, Lebesgue measurable in $t$ for 
      each fixed $x$ and
      continuous in $x$ for each fixed $t$. Suppose that
      for each compact set $U$ in $R$, there
      is a Lebesgue integrable function $m_U(t)$ such that
      for all $(t,x),(t,y)\in U$,
      \[
        |f(t,x)-f(t,y)|<m_U(t)|x-y|.
      \]
      Then for any $(t_0,x_0)\in U$ there exists a unique solution
      to the system \eqref{ODE}.
    \end{ThmI.5.3}

    Since $f$ and $m$ satisfy the conditions of the theorem and as
    \begin{align*}
      |f(t,x)-f(t,y)|&=
        \sum_{i}\lvert 
          \sum_{j,k}\hat\gamma^j(t)\Gamma_{jk}^i(x^k-y^k)\rvert 
        \leq\sum_{i,j,k}\lvert \hat\gamma^j(t)\Gamma_{jk}^i\rvert 
          \lvert x^k-y^k\rvert \\
        &\leq\sum_{i,j,k}\lvert \hat\gamma^j(t)\Gamma_{jk}^i\rvert 
          \lvert x^k-y^k\rvert +
          \sum_{i,j,k,l\neq k}\lvert \hat\gamma^j(t)\Gamma_{jl}^i\rvert 
          \lvert x^k-y^k\rvert \\
        &=\sum_{i,j,k,l}\lvert \hat\gamma^j(t)\Gamma_{jk}^i\rvert 
          \lvert x^l-y^l\rvert \\
        &=\left(\sum_{i,j,k}|\hat\gamma^j(t)\Gamma_{jk}^i|\right)
          \left(\sum_l|x^l-y^l|\right)\\
        &=\frac{m(t)}{c}|x-y|
        < m(t)|x-y|,
    \end{align*}
    holds on all of $R$,
    where we have used our requirement that $c>1$, our
    solution, $X^i:I_\phi\to\mathbb{R}$, $i=1,\ldots,n$, is unique.
    Hence for $i=1,\ldots,n$ there exists 
    a unique and absolutely continuous 
    function $X^i:I_\phi\to\mathbb{R}$ so that
    $X^i(\tau)=v^i$ and
    \[
      \frac{d}{d\tau}X^i(t)+(\gamma')^j(t)X^k(t)\Gamma_{jk}^i(t)=0
    \]
    for almost all $t$. 
    
    Since $\man{M}$ is paracompact and Hausdorff and as 
    $[a,b)\subset\mathbb{R}$ 
    we know that there exists a covering of $[a,b)$ by 
    a countable collection of compact connected intervals
    $\{I_i\subset [a,b):i\in\mathbb{N}\}$ so that 
    $a\in I_1$,
    $I_i\cap I_{i+1}\neq \EmptySet$, $I_{i-1}\cap I_i\neq\EmptySet$ and
    $I_i\cap I_j=\EmptySet$ if $j\neq i, i\pm 1$ and so that for each $i$
    there exists a chart $\phi_i$ so that $\gamma(I_i)$ lies in $\phi_i$.
    Without loss of generality we assume that $\phi=\phi_l$ and that 
    $I_\phi = I_l$ for some $l\in\mathbb{N}$.
          
    Suppose that there exist $X^k:I_i\to\mathbb{R}$, $k=1,\ldots,n$,
    absolutely continuous functions satisfying
    \eqref{ODE} except on a set of Lebesgue measure zero.
    As $I_i\cap I_{i+1}$ has non-zero Lebesgue measure there exists
    $t_{i,i+1}\in I_i\cap I_{i+1}$ such that each $X^k$ exists at
    $t_{i,i+1}$. From above, 
    there exists, except on a set of Lebesgue measure zero,
    an absolutely continuous solution $Y^k:I_{i+1}\to\mathbb{R}$
    of \eqref{ODE}
    so that for each $k=1,\ldots,n$, we have $Y^k(t_{i,i+1})=X^k(t_{i,i+1})$.
    By uniqueness of the solution we must have 
    $Y^k=X^k$ on $I_i\cap I_{i+1}$
    and therefore we can extend our solution to $I_i\cup I_{i+1}$ in an
    absolutely continuous way. The same argument can be used to 
    show that the solution on $I_i$ can be extended to $I_{i-1}\cup I_i$.
    By induction, and as we have proven existence
    on $I_l$, we can extend our solution to
    all of $[a,b)$. We denote the vector that
    results from this extension by $V$.
    
    Therefore, we have that
    $V$ is a vector field on $\gamma$, which is
    absolutely continuous on compact subsets of $[a,b)$ 
    and so that $V(\tau)=v$
    and $\nabla_{\gamma'}V=0$ except on a set of Lebesgue 
    measure zero, as required.
  \end{proof}

  Given two such parallelly
  propagated vector fields $X,Y$ on a locally Lipschitz curve
  $\gamma:[a,b)\to \man{M}$ then 
  $g(X,Y):[a,b)\to \mathbb{R}$ is a function 
  that is absolutely continuous on compact subsets of $[a,b)$
  and
  $\gamma'(g(X,Y))=0$ almost
  everywhere. From the comments just before Proposition 9.6.4, from
  Proposition 9.6.4 itself and from Proposition 9.6.6 of
  \cite{HaaserSullivan1991} the function $g(X,Y)$ is constant on
  $[a,b)$. Thus, if $X$ and $Y$ are chosen so that
  for some $t\in [a,b)$, $g(X,Y)(t)=c$ then
  $g(X,Y)(\tau)=c$ for all $\tau\in [a,b)$.
  
  Using similar arguments it is possible to show that 
  this definition of parallel propagation on locally Lipschitz curves
  satisfies all the expected properties.

  \begin{Def}\label{continuousGAP}
   Let $\gamma:[a,b)\to \man{M}$ be a locally Lipschitz curve, 
   choose
   $c\in\mathbb{R}$,
   $t_0\in [a,b)$ and let
   $X_1,\ldots,X_n$ be a frame of
   linearly independent 
   vector fields on $\gamma$ which are absolutely continuous
   on compact subsets of $[a,b)$ and so that
   $\nabla_{\gamma'}X_i=0$ almost everywhere. Then we may write
   $\gamma'=(\gamma')^iX_i$ for almost all $t\in [a,b)$. Since
   $\gamma$ is locally Lipschitz the 
   function $\tau:[a,b)\to\mathbb{R}$ given by,
   \[
     \tau(t)=
     \bigint_{\displaystyle t_0}^{\displaystyle t}
       \sqrt{\sum_{i=1}^n\left((\gamma')^i(u)\right)^2}\,\textrm{d} u
       \ensuremath{ + c},
   \]
   is well defined.
   We shall call $\tau$ a generalized affine parameter.
  \end{Def}

  It is clear that Definition \ref{continuousGAP} when
  applied to $C^1$ curves reproduces the standard generalized affine
  parameter. Thus Definition \ref{continuousGAP} 
  is a generalization of the generalized affine parameter for locally Lipschitz
  curves.

  Unlike general locally Lipschitz changes of parameter,
  any locally Lipschitz curve reparametrized with a generalized affine
  parameter remains locally Lipschitz. This is because,
  if $\gamma$ is a locally Lipschitz curve and $\mu$ is $\gamma$
  reparametrized with the generalized affine parameter $\tau$, 
  i.e.\ $\mu(\tau(t))=\gamma(t)$, then
  the length of $\mu'(\tau)$ is $1$ in the Riemannian 
  metric induced by the frame used to define $\tau$. Restricting to
  relatively compact subsets of the domain of $\tau$ and translating
  this result into a coordinate frame demonstrates that $\mu$ is
  locally Lipschitz.

  Since, by assumption, every continuous causal curve
  is also locally Lipschitz, continuous causal curves can 
  be given generalized affine parameters.

  \begin{Def}
    Let $\mathcal{C}_{\text{ll}}(\man{M})$ be the set of 
    all locally Lipschitz curves in $\man{M}$ with generalized affine parameters.
    Let $\mathcal{C}_{\text{cc}}(\man{M})$ be the set of all continuous
    causal curves in $\man{M}$ with generalized affine parameters. 
  \end{Def}
    
    By Definition \ref{def.ccurve},
    $\mathcal{C}_{\text{cc}}(\man{M})\subset\mathcal{C}_{\text{ll}}(\man{M})$.

  \begin{Prop}
    The sets $\mathcal{C}_{\text{ll}}(\man{M})$ 
    and $\mathcal{C}_{\text{cc}}(\man{M})$ satisfy the bounded parameter
    property.
  \end{Prop}
  \begin{proof}
    Through any point $p\in \man{M}$ there exists a causal geodesic. Such a
    geodesic is a continuous causal curve and therefore an element of
    $\mathcal{C}_{\text{cc}}(\man{M})$, hence also of
    $\mathcal{C}_{\text{ll}}(\man{M})$.
    It is clear that any subcurve of a generalized affinely parametrized
    locally Lipschitz curve is also a generalized affinely parametrized
    locally Lipschitz
    curve. Likewise, the definition of a continuous causal curve makes it clear
    that a subcurve of a 
    generalized affinely parametrized
    continuous causal curve is also a 
    generalized affinely parametrized
    continuous causal
    curve.
    Since the inner products of parallelly propagated vectors along locally
    Lipschitz
    curves are constant, the standard proof \cite[Page 259]{HawkingEllis1973}
    that if one generalized
    affine parameter (on $C^1$ curves) is bounded then every generalized affine
    parameter (on $C^1$ curves) is bounded carries over to locally 
    Lipschitz curves.
    Thus,
    if two generalized affinely
    parametrized locally Lipschitz curves are related by a change of parameter
    then either both parameters are bounded or both are unbounded. Since
    generalized affinely parametrized
    continuous causal curves are generalized affinely parametrized
    locally Lipschitz curves this holds for generalized affinely parametrized 
    continuous
    causal curves as well.
    Hence, the bounded parameter property holds on both sets.
  \end{proof}

  Thus, we have identified a bounded parameter property satisfying set of
  curves that is larger than the set of all affinely parametrized causal
  geodesics and includes all continuous causal curves.

  \subsection{Incompleteness for locally Lipschitz curves}

  We now give a definition of incompleteness 
  for locally Lipschitz curves and show that
  locally Lipschitz curves with endpoints are not incomplete. 
  To do this we mirror
  Definition \ref{Def:incomplete_geo}.

  \begin{Def}
    A generalized affinely parametrized locally 
    Lipschitz curve $\gamma:[a,b)\to \man{M}$
    is incomplete if $b<\infty$ and $\gamma$ is not
    extendible
    by any generalized affinely parametrized locally Lipschitz curve.
  \end{Def}

  The following technical lemma will be of use below.

  \begin{Lem}\label{Lem:extensionofLL}
    Let $\gamma:[a,b)\to \man{M}$ be a generalized affinely parametrized 
    locally Lipschitz
    curve and
    let
    $\lambda:[a',b')\to \man{M}$ be a locally Lipschitz curve
    such that there exists
    a change of parameter 
    $f:[a,b)\to[a',c')$, $c'\leq b'$, so that $\lambda\circ f = \gamma$.
    Then there exists
    $s:[a',b')\to[a,d)$, $b\leq d$, a generalized affine parameter on
    $\lambda$, so that $s\circ f$ is the identity on $[a,b)$,
    and 
    $c'<b'$ if and only if $b<d$.
  \end{Lem}
  \begin{proof}
    Since $\gamma$ is generalized affinely parametrized there exists
    $t_0\in[a,b)$ and  
    a parallelly propagated frame $X_1,\ldots,X_n$ on $\gamma$ so that
    the parameter $t\in[a,b)$ is given by
    \[
      t=\bigint_{\displaystyle t_0}^{\displaystyle t}
      \sqrt{\sum_{i=1}^n{((\gamma')^i(\hat t\hspace{1pt}))^2}}\,\textrm{d}\hat t + t_0
    \]
    where $\gamma'=(\gamma')^iX_i$. 
    Since $\gamma(t)=\lambda\circ f(t)$ and $\lambda$ is locally Lipschitz
    the equation $\gamma'(t)=f'(t)\lambda'(f(t))$ holds almost everywhere.
    Therefore,
    \[
      0=\nabla_{\gamma'}X_i=f'\nabla_{\lambda'(f)}X_i.
    \]
    Hence, $\nabla_{\lambda'(f)}X_i=0$ almost everywhere and so the frame
    $X_1,\ldots,X_n$  can be extended to
    a parallelly propagated frame along $\lambda$.
    Let $\lambda'=(\lambda')^iX_i$ and
    let
    $s_0=f(t_0)$.
    This allows us to define a
    generalized affine
    parameter, $s$, on $\lambda$ by
    \[
      s(\tau_s)=\bigint_{\displaystyle s_0}^{\displaystyle \tau_s}
      \sqrt{\sum_{i}^n{((\lambda')^i(\hat s\hspace{1pt}))^2}}
      \,\textrm{d}\hat s + t_0,\quad\quad
      \tau_s\in[a',b').
    \]
    Choose $c,d\in\mathbb{R}\cup\{\infty\}$, $c<d$, so that
    $s:[a',b')\to[c,d)$ is surjective.

    Let $\tau_t\in [a,b)$ and let $\tau_s=f(\tau_t)$.
    Since $f$ is locally Lipschitz we know that $f$ is absolutely continuous on
    the interval $[t_0,\tau_t]$.
    Hence, \cite[Theorem 9.7.5]{HaaserSullivan1991} implies that we
    can use $f$ to change variables in the equation for $s$.
    From above $(\gamma')^i(t)=f'(t)(\lambda')^i(f(t))$ and therefore
    \begin{align*}
      s(\tau_s)&=\bigint_{\displaystyle s_0}^{\displaystyle \tau_s}
      \sqrt{\sum_{i}^n{((\lambda')^i(\hat s\hspace{1pt}))^2}}\,\textrm{d}\hat s + t_0
      =\bigint_{\displaystyle t_0}^{\displaystyle \tau_t}
      \sqrt{\sum_{i}^n{((\lambda')^i(f(\hat t\hspace{1pt})))^2}}\,f'(\hat t)
      \,\textrm{d}\hat
      t + t_0\\[6pt]
      &=\bigint_{\displaystyle t_0}^{\displaystyle \tau_t}
      \sqrt{\sum_{i}^n{((\gamma')^i(\hat t\hspace{1pt}))^2}}
      \,\textrm{d}\hat t + t_0= \tau_t.
    \end{align*}
    Hence, $s\circ f(\tau_t)=\tau_t$.
    Since $\tau_t$ was arbitrary
    this implies that $s\circ f(t) = t$ for all $t\in[a,b)$,
    as required.

    As $f$ is strictly monotonically
    increasing and surjective 
    we have that $f(a)=a'$. Since $\lambda$ is locally
    Lipschitz the function
    \[
      s'(\tau_s)=\sqrt{\sum_{i}^n{((\lambda')^i(\tau_s))^2}}
    \]
    is positive almost everywhere. Thus $s$ is strictly monotonically
    increasing. Hence, $s(a')=c$ and therefore, as $s\circ f(a)=a$, $c=a$.
    So $s:[a',b')\to[a,d)$ as required.

    Let $\epsilon>0$ be such that $b-\epsilon\in [a,b)$ then
    $s\circ f(b-\epsilon)\in[a,d)$. Hence $b-\epsilon < d$ for all
    such $\epsilon$. This implies that $b\leq d$, as required.

    Suppose that $c'<b'$ and, again,  
    let $\epsilon>0$ be such that $b-\epsilon\in[a,b)$. 
    Then $f(b-\epsilon)\in[a',c')$ and thus $f(b-\epsilon)<c'$. 
    Since $c'<b'$, $s(c')$ is defined.
    Hence $b-\epsilon=s\circ f(b-\epsilon)<s(c')$. As this is true
    for all such 
    $\epsilon>0$ we have that $b\leq s(c')$. Since $s(c')\in[a,d)$
    we see that $b<d$ as required.

    Suppose that $b<d$. Let $\epsilon>0$
    be such that $c'-\epsilon\in[a',c')$.
    Since $f$ is surjective there exists $x\in[a,b)$ so that
    $f(x)=c'-\epsilon$. Since $b<d$, $s^{-1}(b)$ is well defined.
    As $x\in[a,b)$ we have $x<b$ and hence $s^{-1}(x)<s^{-1}(b)$.
    By construction
    \[
      s^{-1}(x)=s^{-1}\circ s\circ f(x) = s^{-1}\circ
      s(c'-\epsilon)=c'-\epsilon.
    \]
    Thus $c'-\epsilon<s^{-1}(b)$. As this holds for all 
    such $\epsilon>0$, $c'\leq
    s^{-1}(b)$. Since $s^{-1}(b)\in[a',b')$ we see that $c'\leq s^{-1}(b)<b'$
    as required.
  \end{proof}

  \begin{Prop}\label{prop:llextend}
    Let $\gamma:[a,b)\to \man{M}$ be a generalized affinely parametrized 
    locally
    Lipschitz curve. If $\gamma$ has an endpoint then it is extendible
    and $b<\infty$.
  \end{Prop}  
  \begin{proof}
    If $b<\infty$, let $f:[a,b)\to[0,1)$ be defined by 
    $f(t)=\frac{1}{b-a}(t-a)$. 
    If
    $b=\infty$, let $f:[a,b)\to[0,1)$ be defined by
    $f(t)=\frac{2}{\pi}\arctan(t-a)$. In either case both
    $f$ and $f^{-1}$ are changes of
    parameter. 

    Let $p\in \man{M}$ be such that $\gamma\to p$.
    Let $\lambda:[0,c)\to \man{M}$, $c\in\mathbb{R}^+$, be a geodesic from $p$. Let
    $\mu:[0,c+1)\to \man{M}$ be defined by
    \[
      \mu(t)=\left\{\begin{aligned}
          \gamma\circ f^{-1}(t) && t\in[0,1)\\
              \lambda(t-1) && t\in[1,c+1).
        \end{aligned}\right.
    \]
    Since $1<c+1$,
    by Lemma \ref{Lem:extensionofLL} there exists 
    $d\in\mathbb{R}\cup\{\infty\}$,
    $b<d$,
    and a generalized affine parameter
    $s:[0,c+1)\to[a,d)$ on $\mu$ so that
    $s\circ f(t)=t$ for all $t\in[a,b)$.
    Let $\delta:[a,d)\to \man{M}$ be the generalized affinely parametrized
    curve defined by $\delta\circ s = \mu$.
    Then $\delta =\delta\circ s\circ f = \mu\circ f = \gamma$ on $[a,b)$.
    Since $b<d$, $b<\infty$ and the curve
    $\delta$
    is an extension of $\gamma$.
  \end{proof}

  It is possible to extend continuous causal curves with endpoints by
  continuous causal curves.

  \begin{Prop}\label{CurLim:Prop.WHasOneElementIFFGammaIsExtendibleAlt}
    Let $\gamma:[a,b)\to\man{M}$ be a generalized
    affinely parametrized continuous future directed 
    non-spacelike (timelike) curve.
    If $\gamma$ has an endpoint then $\gamma$ is extendible by
    a generalized affinely parametrized continuous future directed
    non-spacelike (timelike) curve and $b<\infty$.
  \end{Prop}
  \begin{proof}
    From Proposition \ref{prop:llextend}
    it is clear that $b<\infty$ and there exists
    a generalized affinely parametrized locally Lipschitz curve 
    $\delta$ that is an extension of $\gamma$, where $\gamma$
    is considered as a locally Lipschitz curve.
    To prove the result it remains to
    show that $\delta$, as defined in the proof
    of Proposition \ref{prop:llextend}, 
    is a continuous future directed non-spacelike
    (timelike) curve. Note that, by construction
    $\delta:[a,d)\to\man{M}$, $b<d$ and
    $\delta = \gamma$ on $[a,b)$. 

    By construction $\delta$ is continuous. 
    Since $\delta|_{[a,b)}$ is a
    continuous, future directed, non-spacelike (timelike) curve and
    $\delta|_{[b,d)}$ can be chosen to be a smooth, future directed, 
    non-spacelike (timelike) geodesic,
    we know that for $\delta$ 
    to be a non-spacelike, future directed, continuous curve we need only 
    show that
    there exists a neighbourhood $N\subset [a,d)$ of 
    $b$ and a convex normal neighbourhood $U$ of $p=\delta(b)$  
    so that for all $t\in N$, $t\neq b$, if $t>b$ then 
    $\delta(t)\in J^+(p,U)-p$
    and
    if $t<b$ then $\delta(t)\in J^-(p,U)-p$. 
    In the timelike case we need $N$ and $U$ as before 
    so that for all $t\in N$, $t\neq b$, if $t>b$ then 
    $\delta(t)\in I^+(p,U)$
    and
    if $t<b$ then $\delta(t)\in I^-(p,U)$.

    Choose $U$ a convex normal neighbourhood of $p$ and let $N$
    be a neighbourhood of $b$ in $[a,d)$
    so that $\delta(N)\subset U$. Let $t\in N$.
    Suppose that $t>b$. 
    Then the condition immediately follows as $\delta|_{[b,d)}$ is a 
    non-spacelike (timelike), future directed geodesic.

    Suppose that $t<b$ and that $\gamma$ is non-spacelike. As $U$ is a
    convex normal neighbourhood
    there exists a unique geodesic $\alpha:[0,1]\to U$ so that
    $\alpha(0)=\delta(t)$ and $\alpha(1)=p$. If $\alpha$ is future directed
    non-spacelike then we are done, so suppose 
    that $\alpha$ is spacelike. 
    Then by the continuity of $g$ on $T_{\delta(t)}\man{M}$ and Proposition 4.5.1
    of \cite{HawkingEllis1973} there exists
    a neighbourhood $V$ of $p$ in $U$ 
    so that for all $q\in V$ the
    unique geodesic in $U$ from $\delta(t)$ to $q$ is spacelike. Since
    $\delta(\tau)\to p$ as $\tau\to b$, we can consider all $t'<b$,
    where $t<t'$, such that $\delta(t')\in V$. By construction,
    for each such $t'$, there is a unique spacelike 
    geodesic $\tilde{\alpha}_{t'}:[0,1]\to U$ so that
    $\tilde{\alpha}_{t'}(0) = \delta(t)$ and
    $\tilde{\alpha}_{t'}(1) = \delta(t')$.
    By Proposition 4.5.1 of
    \cite{HawkingEllis1973}, however, this implies that for
    each such $t'$, $\delta(t)\not\in J^-(\delta(t'), U)$. This is
    a contradiction and so $\alpha$ must be non-spacelike.
    Also each $\tilde{\alpha}_{t'}$ must be non-spacelike and
    future directed since $\delta(t)\in J^-(\delta(t'), U)$.
    By continuity $\alpha(t)=\lim_{t'\to b}\tilde{\alpha}_{t'}(t)$
    and so $\alpha$ is future directed.
    Lastly since $U$ is normal and as $t<t'$, $\delta(t)\neq\delta(t')$.
  
    Suppose that $t<b$ and that $\gamma$ is timelike. Let $\alpha$ be
    as above. If $\alpha$ is future directed, timelike then we are done, so
    suppose that
    $\alpha$ is non-timelike. If $\alpha$ is spacelike the same argument
    as above can be used to find a contradiction. Thus we
    can restrict here to the case that $\alpha$ is a
    null geodesic. 
    We may choose $t'\in (t,b)$, then by assumption 
    $\delta(t')\in I^+(\delta(t),U)$. This implies that the unique geodesic
    from $\delta(t')$ to $p$ is spacelike (otherwise we would have a
    timelike curve from $\delta(t)$ to $p$ in $U$
    which
    contradicts the assumption that $\alpha$ is null
    \cite[Proposition 4.5.1]{HawkingEllis1973}). The same argument as above 
    can then be applied to find a contradiction. Therefore $\alpha$ must be
    timelike and future directed. 

    Thus $\delta$ is a future directed continuous non-spacelike (timelike) curve
    as required.
  \end{proof}

\section{Proofs of the generalizations}

  \begin{proof}[Proof of Theorem
    {\ref{CurLim:Thm.NewAbstractBoundarySingularitTheorem}}]
    $\Leftarrow$
    Unchanged from Theorem \ref{thm.abst_original}.
    
    $\Rightarrow$ Let
	    $\gamma\in\mathcal{C}$ be incomplete.
	    We
	    have two cases:

    \begin{description}
      \item[\textbf{Case 1}]
        Unchanged from Theorem \ref{thm.abst_original}.
      	
      \item[\textbf{Case 2}] 
        By assumption $\gamma$ is precompact. 
        Then either $\gamma$ has an endpoint or $\gamma$ does not
        have an endpoint.
        \begin{description}
          \item[\textbf{Case 2.1}] 
            Suppose that $\gamma$ has an endpoint. By Proposition 
            \ref{CurLim:Prop.WHasOneElementIFFGammaIsExtendibleAlt}
          	we know that $\gamma$ is extendible.
            This is a contradiction
            and thus this case cannot occur.
          \item[\textbf{Case 2.2}]
            By assumption $\gamma$ is a precompact, winding curve.
            Since $\gamma$ is winding it has two full sequences 
            with different limit
            points. As $\man{M}$ is Hausdorff this implies that $\gamma$ is
            inextendible.
            Thus the inextendible curve
            $\gamma$ is totally imprisoned
            in the compact set 
            $\overline{\gamma}$. 
            Proposition \ref{prop.hawkellis} implies that
            the future (past) distinguishing condition fails on 
            $\overline{\gamma}$.
            This is a
            contradiction and therefore this case cannot occur.
        \end{description}
    \end{description}

    As only case 1 may occur, we have proven our
    result.
  \end{proof}

  \begin{proof}[Proof of Theorem
    {\ref{CurLim:Thm.NewAbstractBoundarySingularitTheoremLL}}]
    $\Leftarrow$
    Unchanged from Theorem \ref{thm.abst_original}.
    
    $\Rightarrow$ Let
	    $\gamma\in\mathcal{C}$ be incomplete.
	    We
	    have two cases:

    \begin{description}
      \item[\textbf{Case 1}]
        Unchanged from Theorem \ref{thm.abst_original}.
      	
      \item[\textbf{Case 2}] 
        By assumption $\gamma$ is precompact. 
        Then either $\gamma$ has an endpoint or $\gamma$ does not
        have an endpoint.
        \begin{description}
          \item[\textbf{Case 2.1}] 
            Suppose that $\gamma$ has an endpoint. By Proposition 
            \ref{prop:llextend}
          	we know that $\gamma$ is extendible.
            This is a contradiction
            and thus this case cannot occur.
          \item[\textbf{Case 2.2}]
            By assumption $\gamma$ is a precompact, winding curve.
            This implies that
            $\gamma$ is inextendible and totally imprisoned.
            By assumption such curves cannot exist.
            This is a
            contradiction and therefore this case cannot occur.
        \end{description}
    \end{description}

    As only case 1 may occur, we have proven our
    result.
  \end{proof}

\end{document}